\documentclass[letterpaper, 10 pt, conference]{ieeeconf} 
\IEEEoverridecommandlockouts
\usepackage{amsmath,amssymb}
\usepackage{graphicx}
\usepackage{bm}
\usepackage{svg}
\usepackage{textcomp}
\usepackage{xcolor}
\usepackage{cite}
\usepackage{algorithm}
\usepackage[noend]{algorithmic}
\usepackage{xprintlen}
\usepackage{cases}
\usepackage{tabularx,ragged2e}
\usepackage{booktabs}
\usepackage{nccmath}
\usepackage{caption}
\usepackage{float}
\usepackage{subcaption}
\usepackage{mathtools}
\def\BibTeX{{\rm B\kern-.05em{\sc i\kern-.025em b}\kern-.08em
    T\kern-.1667em\lower.7ex\hbox{E}\kern-.125emX}}
\usepackage{array}
\newcolumntype{P}[1]{>{\centering\arraybackslash}p{#1}}
\renewcommand{\vec}[1]{\bm{#1}}
\newcommand{\true}{\top}
\newcommand{\false}{\bot}

\newtheorem{problem}{\textbf{Problem}}
\newtheorem{proposition}{\bf{Proposition}}
\newtheorem{definition}{\bf{Definition}}
\newtheorem{lemma}{\bf{Lemma}}
\newtheorem{theorem}{\bf{Theorem}}
\newtheorem{remark}{\bf{Remark}}
\newtheorem{assumption}{\bf{Assumption}}

\newtheorem{corollary}{\bf{Corollary}}

\def\myunderbar#1{\underline{\sbox\tw@{$#1$}\dp\tw@\z@\box\tw@}}

\begin{document}

\title{Decentralized Control of Multi-Agent Systems Under Acyclic \\ Spatio-Temporal Task Dependencies}

\author{Gregorio Marchesini, Siyuan Liu, 
Lars Lindemann and Dimos V. Dimarogonas
	\thanks{This work was supported in part by the Horizon Europe EIC project SymAware (101070802), 
the ERC LEAFHOUND Project, the Swedish Research Council (VR), Digital Futures, and the Knut and Alice Wallenberg (KAW) Foundation.}
\thanks{Gregorio Marchesini, Siyuan Liu, and Dimos V. Dimarogonas are with the Division
of Decision and Control Systems, KTH Royal Institute of Technology, Stockholm, Sweden.
	E-mail: {\tt\small \{gremar,siyliu,dimos\}@kth.se}.  Lars Lindemann is with the Thomas Lord Department of Computer Science, University of Southern California, Los Angeles, CA, USA.
	E-mail: {\tt\small \{llindema\}@usc.ed}.   
}
}

\maketitle
\begin{abstract}
We introduce a novel distributed sampled-data control method tailored for heterogeneous multi-agent systems under a global spatio-temporal task with acyclic dependencies. Specifically, we consider the global task as a conjunction of independent and collaborative tasks, defined over the absolute and relative states of agent pairs. Task dependencies in this form are then represented by a task graph, which we assume to be acyclic. From the given task graph, we provide an algorithmic approach to define a distributed sampled-data controller prioritizing the fulfilment of collaborative tasks as the primary objective, while fulfilling independent tasks unless they conflict with collaborative ones. Moreover, communication maintenance among collaborating agents is seamlessly enforced within the proposed control framework. A numerical simulation is provided to showcase the potential of our control framework.
\end{abstract}


\section{Introduction}
Over the past decade, Control Barrier Functions (CBFs) have emerged as versatile control solutions in the realm of cyber-physical systems. Indeed, CBF-like constraints have enabled the seamless integration of multiple control objectives into optimal control frameworks, such as model predictive control and quadratic programming-based controllers, as well as analytical feedback control laws \cite{lindemann2018control,zehfroosh2022non,chen2020guaranteed,singletary2021comparative,distributedFormationFuJunie}.\par
This work is concerned with the development of a novel distributed sampled-data controller applicable to heterogeneous multi-agent systems subject to communication and spatio-temporal constraints expressed as Signal Temporal Logic (STL) specifications. For multi-agent systems, \cite{panagou2015distributed,larsTCNS,nonsmoothBarriers,distributedSecondOrder,wang2016safety} have developed distributed/decentralized control laws meeting several objectives, such as formation control and stabilization, as well as connectivity and collision avoidance constraints. Nevertheless, the interplay between connectivity, safety, and spatio-temporal task fulfilment 
requires further investigation. In \cite{larsTCNS,nonsmoothBarriers}, the proposed CBF-enabled distributed control law requires global knowledge of the entire system state, although the control inputs are computed locally. Hence, the framework is only applicable to systems with all-to-all communication. In \cite{panagou2015distributed}, a novel decentralized control law merging stabilization and safety objectives into Lyapunov-like barrier functions for controlling a swarm of differential drive robots was proposed. In their approach, agents are assumed to be connected to a central broadcasting agent from which goal destinations are dispatched, while collisions are resolved by a local avoidance scheme. The works in \cite{distributedSecondOrder, wang2016safety} offer a decentralized formation control approach with input limitations and collision/connectivity constraints satisfied through suitable CBFs; however, more general objectives, like spatio-temporal tasks, are not considered therein. At the same time, none of the previous approaches is designed to safely handle control inputs expressed in a zero-order hold fashion, as typically applied from real embedded controllers.\par
Given these previous results, we propose a quadratic programming-based decentralized control framework that incorporates STL specification (tasks), as well as communication constraints in the form of CBF-like constraints. Building upon  \cite{roque2022corridor,breeden2021control,shawPaper}, we provide a definition of sampled-data control barrier functions, which are applied to encode the control objectives and provide continuous-time guarantees over their satisfaction. Furthermore, following our previous work \cite{marchesini2024communication}, we assume a global STL task assigned to the MAS can be decomposed as a conjunction of independent and collaborative tasks defined over the absolute and relative state of couples of agents in the system, respectively, from which an acyclic task graph is derived. The established acyclic task graph is then applied to define a leader-follower structure over each edge of the graph, such that leaders are responsible for the satisfaction of collaborative tasks over their assigned edge, while followers reduce their maximum control authority to favour the satisfaction of the former task. In summary, the main contribution of this work is the development of a novel sampled-data decentralized control framework that guarantees continuous-time satisfaction of spatio-temporal tasks with acyclic dependencies expressed as STL specifications, as well as communication constraints. The presentation is organised as follows: Sec \ref{Preliminaries} summarizes preliminaries and problem formulation. In Sec. \ref{sample data section} we define sampled-data CBFs and how STL tasks can be encoded in this framework. In Sec. \ref{control}, the novel decentralized control scheme is presented together with relevant algorithms for its implementation, while Secs. \ref{simulations}-\ref{conclusions} provide numerical simulations and conclusions.\par
\textit{Notation}: Bold letters indicate vectors while capital letters indicate matrices. Vectors are considered to be column vectors. The notation $|\mathcal{A}|$ indicates the cardinality of the set $\mathcal{A}$. For a given scalar $\gamma \in \mathbb{R}$, let $\gamma \mathcal{A}:=\{\gamma a | a \in \mathcal{A}\}$. Let $\oplus$ and $\ominus$ indicate the Minkowski sum and difference of two sets. Let $blk(A_1,\ldots A_n)$ represent the block diagonal matrix with blocks given by the matrices $A_1,\ldots A_n$. We use the notation $\partial_{x} := \frac{\partial}{\partial x}$ to identify the partial derivative w.r.t to $x$. The set $\mathbb{R}_{+}$ denotes the non-negative real numbers while the set $\mathbb{N}_0:=\mathbb{N} \cup \{0\}$ is the set of natural numbers including zero.

\section{Preliminaries and problem formulation}\label{Preliminaries}
Let $\mathcal{V}=\{1,\ldots N\}$ be the set of indices assigned to each agent in a multi-agent system and let each agent be governed by an input-affine dynamics of the form:
\begin{equation}\label{eq:single agent dynamics}
\dot{\vec{x}}_i    = f_i(\vec{x}_i) + g_i(\vec{x}_i)\vec{u}_i, 
\end{equation}
with state $\vec{x}_i\in \mathbb{R}^n \subset \mathbb{X}_i$ and input vector $\vec{u}_i\in \mathbb{R}^{m_i} \subset \mathbb{U}_i$, where $m_i$ is the input dimension for agent $i$. Let $\mathbb{X}_i$ be compact and $\mathbb{U}_i$ be compact and convex. Moreover, let the functions $f_i : \mathbb{R}^n \rightarrow  \mathbb{R}^n$, $g_i : \mathbb{R}^n \rightarrow \mathbb{R}^{n\times m_i}$ be locally Lipschitz continuous over $\mathbb{X}_i$. The global MAS dynamics is then compactly written as
\begin{equation}\label{eq:multi agent dynamics}
    \dot{\vec{x}} = \bar{f}(\vec{x}) + \bar{g}(\vec{x})\vec{u}
\end{equation}
where $\vec{x}:= [\vec{x}^T_1,\ldots \vec{x}^T_{N}]^T, \vec{u}:= [\vec{u}^T_1,\ldots \vec{u}^T_{N}]^T, \bar{f}(\vec{x}):=[f_1^T, \ldots f_{N}^T]^T$ and $\bar{g}(\vec{x}):=blk(g_1,\ldots g_N)$.
Let $\vec{e}_{ij}:=\vec{x}_i-\vec{x}_{j} \in \mathbb{X}_{ij}$ represent the relative state vector for each $i,j\in\mathcal{V}$ where $\mathbb{X}_{ij}:= \mathbb{X}_i \ominus \mathbb{X}_j$. Also let $\vec{p}_i = S\vec{x}_i$ represent the position of each agent $i$ for a given selection matrix\footnote{A selection matrix applies to select $m\leq n$ unique elements from a vector of dimensions $n$ such that $S\in \mathbb{R}^{m \times n}$ with $S[i,j] \in \{0,1\}$;  $\sum_{j=1}^{n}S[i,j] =1, \, \forall i=1,\ldots m$ and $\sum_{i=1}^{m}S[i,j] =1, \, \forall j=1,\ldots n$.} $S$ and let the relative position vector $\vec{p}_{ij}:=\vec{p}_i-\vec{p}_j$ and the communication radius $r_c>0$ such that it is assumed that at time $t$ agents $i$ and $j$ can collaborate toward the satisfaction of a collaborative task (defined in Sec. \ref{stl section}) if $\|\vec{p}_{ij} \|\leq  r_c$.
\begin{figure}
    \centering
    \includegraphics[width=0.5\linewidth]{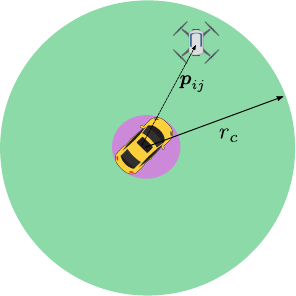}
    \caption{Two heterogeneous agents (a drone and a ground vehicle with differential drive dynamic) within their respective communication radius}
    \label{fig:enter-label}
\end{figure}

\subsection{Signal Temporal Logic and Task Graph}\label{stl section}
Signal Temporal Logic is a predicate logic defined over a set of boolean-valued predicates $\mu(h(\vec{x})):=
\bigl\{\begin{smallmatrix*}
\true &\text{if} \; h(\vec{x})\geq 0 \\
\false &\text{if} \; h(\vec{x})<0 ,
\end{smallmatrix*}$ where $h: \mathbb{R}^{n\cdot N}\rightarrow \mathbb{R}$ is a real-valued \textit{predicate function}. Intuitively, $h(\vec{x})$ encodes some spatial constraints for the system. In the following, we consider predicate functions of the form $h_{i}:= h_i(\vec{x}_i)$ or $h_{ij}:= h_{ij}(\vec{e}_{ij})$ for  $i,j\in \mathcal{V}$. We refer to the former as \textit{independent} predicate function and the latter as \textit{collaborative} predicate function.
\begin{assumption}\label{concavity}
    Individual/Collaborative predicate functions $h_{i}$ and $h_{ij} \; \forall i,j\in \mathcal{V}$ are  concave functions.
\end{assumption}\par
Typical examples of functions that fit into our framework are 1) communication maintenance $h_{ij}(\vec{e}_{ij})= r_c^2-\|\vec{p}_{ij}\|^2$ 2) relative position formation $h_{ij}(\vec{e}_{ij})= r^2-(\vec{p}_{ij}-\vec{c}_{ij})^TP(\vec{p}_{ij}-\vec{c}_{ij})$ with formation vector $\vec{c}_{ij}\in \mathbb{R}^3$, $r>0$ and positive definite matrix $P\in \mathbb{R}^{3\times 3}$,  3) polyhedral formation $h_{ij}(\vec{e}_{ij})= -log \bigl(\sum_{i=1}^p exp(\vec{a}_i^T\vec{e}_{ij}-b_i) \bigr)$ where $\vec{a}_i \in \mathbb{R}^{n}$ and $b_i\in \mathbb{R}$ define the $p$ hyperplanes $\{\vec{x}\in \mathbb{R}^n| \vec{a}_i^T \vec{x} -b_i \geq 0 \}$ for some $p\geq1$, 4) Go-to-goal specifications like $h(\vec{x}_{i})= \epsilon^2-\|\vec{p}_i-\vec{t}_{i}\|^2$ for some $\epsilon\geq0$ with goal position $\vec{t}_i\in \mathbb{R}^3$.
The applied STL fragment is recursively given as
\begin{subequations}\label{eq:working fragment}
\begin{align}
\varphi_{i}  &:= F_{[a,b]}\mu_i|G_{[a,b]}\mu_i\label{eq:single agent spec}\\
\varphi_{ij} &:=  F_{[a,b]}\mu_{ij}|G_{[a,b]}\mu_{ij} \label{eq:multi agent spec}\\
\phi_{ij} &:= \wedge_{k} \varphi^k_{ij},\; \phi_{i}  := \wedge_{k} \varphi^k_{i} ,
\end{align}
\end{subequations}
where $\mu_{ij}:= \mu(h_{ij}(\vec{e}_{ij}))$ and $\mu_{i}:= \mu(h_{i}(\vec{x}_{i}))$. The operators $G$ and $F$ are the \textit{always} and \textit{eventually} operators. For a given STL task $\phi$,  the notation $\vec{x}(t)\models \phi$ indicates that the state trajectory $\vec{x}(t)$ satisfies $\phi$. All the conditions under which $\vec{x}(t)$ satisfies $\phi$ (the STL semantics) are given in \cite{maler2004monitoring,fainekos2009robustness} and are not revised here due to space limitations.  Although, recall that $\vec{x}(t)\models G_{[a,b]}\mu \Leftrightarrow \mu(h(\vec{x}(t))):= \top,\, \forall t\in [a,b] $ and  $\vec{x}(t)\models F_{[a,b]}\mu \Leftrightarrow \exists \tau \in [a,b]\, :\;  \mu(h(\vec{x}(\tau))):= \top$. Consistently with previous nomenclature, we refer to tasks of type  \eqref{eq:single agent spec} and \eqref{eq:multi agent spec} as individual and collaborative tasks respectively.\par 
Collaborative and independent tasks from fragment \eqref{eq:working fragment} naturally induce the definition of an undirected \textit{task graph} $\mathcal{G}_{\psi}(\mathcal{V},\mathcal{E}_{\psi})$ where $\mathcal{E}_{\psi}\subset \mathcal{V}\times \mathcal{V}$ with $(i,j)\in \mathcal{E}_{\psi}$ if there exists a collaborative task $\phi_{ij}$ as per \eqref{eq:multi agent spec} among $i$ and $j$. Due to the limited communication radius $r_c$, we also introduce the time-varying 
 undirected communication graph $\mathcal{G}_c(\mathcal{V},\mathcal{E}_c(t))$ where $(i,j)\in \mathcal{E}_c(t)$ if $\|\vec{p}_{ij}(t)\| \leq r_c$ and let $\mathcal{N}_{\psi}(i)=\{j\in \mathcal{V}| (i,j)\in \mathcal{E}_{\psi}\}$, $\mathcal{N}_{c}(i)=\{j\in \mathcal{V}| (i,j)\in \mathcal{E}_{c}(t)\}$ be the task and communication neighbours for agent $i$. The global task $\psi$ assigned to the MAS is then compactly written as
\begin{equation}\label{eq:global task}
\begin{gathered}
\psi= \psi_{ind} \land \psi_{col}
\end{gathered}
\end{equation} 
\vspace{-0.5cm}\\
where $\psi_{ind}:=\wedge_{i\in \mathcal{V}}  \phi_{i},  \quad \psi_{col} := \wedge_{i\in \mathcal{V}}( \wedge_{j\in \mathcal{N}_{\psi}(i)} \phi_{ij} )$.
\begin{assumption}\label{task symmetry} The task graph $\mathcal{G}_{\psi}$ is acyclic.
\end{assumption}
While Assmp. \ref{task symmetry} can be considered restrictive as it imposes the task dependencies to follow a tree structure, the authors recently proposed a formula rewrite decomposition for tasks in fragment \eqref{eq:working fragment} such that cyclic task graphs can be reformulated into acyclic ones under mild assumption \cite{marchesini2024communication}.
\subsection{Sampled-Data Control}
 Consider a sampling interval $\delta t>0$, initial time $t^0$, sampling instants $t^k:=t^0 + k\delta t$ for $k\in \mathbb{N}_0$ and let the shorthand notation $a^k:=a(t^k)$ for any general quantity $a$ (scalar or vectorial). We consider each agent to be subject to a piece-wise constant (p.w.c) input $\vec{u}_i : \mathbb{R}_+ \rightarrow \mathbb{U}$ such that
\begin{equation}\label{eq:piece-wise constant}
\vec{u}_i(t) = \vec{u}_{i}^k \in \mathbb{U}_i, \quad \forall t\in [t^k,t^{k+1}) \,, \forall i\in \mathcal{V}.
\end{equation}
Furthermore, let $\mathcal{U}_i^{\delta t}$ be the set of p.w.c input $\vec{u}_i(t)$ satisfying \eqref{eq:piece-wise constant} for agent $i$ and  $\mathcal{X}_i$ as the set of state trajectories $\vec{x}_i : \mathbb{R}^+ \rightarrow \mathbb{X}_i$ satisfying \eqref{eq:single agent dynamics} under $\vec{u}_i(t)\in \mathcal{U}^{\delta t}_i$.
\begin{definition}\label{eq:reachable set}
    Consider \eqref{eq:multi agent dynamics} and sampling interval $\delta t>0$. The reachable set from state $\vec{x}^k$ is defined as $\mathcal{R}(\vec{x}^k,\delta t):=\{\vec{x}\in \mathbb{X} | \vec{x}= \int_{t^k}^t\bar{f}(\vec{x}(t)) + \bar{g}(\vec{x}(t))\vec{u}^k\, dt,\;\vec{x}(t^k)=\vec{x}^k,\;  \forall \vec{u}^k\in \mathbb{U},\,  \forall t\in [t^k,t^k + \delta t) \}$.
\end{definition}
\par We hereafter consider that it is possible to compute the reachable set $\mathcal{R}(\vec{x}^k,\delta t)$ for giving initial state $\vec{x}^k$. In the practical application of our proposed control approach, an upper bound of the actual reachable set is sufficient to maintain the soundness of our proposed approach.
\subsection{Problem Formulation}
Given these preliminaries, the problem  approached in this work is formalised as follows:
\begin{problem}\label{main problem}
Given a time-varying communication graph $\mathcal{G}_c(t)$, and acyclic task graph $\mathcal{G}_{\psi}$ associated with the global task $\psi:=\psi_{col}\land \psi_{ind}$ as per \eqref{eq:global task}, define a decentralized p.w.c control law $\vec{u}_i(t)\in \mathcal{U}_i^{\delta t}$ as per \eqref{eq:piece-wise constant} such that 1) $\vec{x}(t)\models \psi$ and 2) $\mathcal{G}_\psi \subseteq \mathcal{G}_c(t),\; \forall t\in \mathbb{R}_+$ if $\mathcal{G}_\psi \subseteq \mathcal{G}_c(t^0)$.
\end{problem}\par 
By the term decentralized we intend that the control input $\vec{u}_i(t)$ at time $t$ for agent $i$ is computed based on the state of the neighbouring agents $j\in \mathcal{N}_{c}(i)$ rather than the global state $\vec{x}$ of the system. Hence, the communication maintenance objective 2) in Problem \ref{main problem} is a precondition for the satisfaction of objective 1). In Theorem \ref{main control theorem} our proposed feedback control law satisfying objective 1) under the assumption that $\mathcal{G}_{\psi} \subseteq \mathcal{G}_{c}(t)$ for all $t\in \mathbb{R}_+$ is provided. On the other hand, Corollary \ref{communicationc corollary} clarifies how the same control solution handles the enforcement of the communication objective $\mathcal{G}_\psi \subseteq \mathcal{G}_c(t),\; \forall t\in \mathbb{R}_+$, by including communication tasks in conjunction with $\psi_{col}$.

\section{Sampled-data time-varying CBFs}\label{sample data section}
Sampled-data Control Barrier Functions (sdCBF) are considered as suitable framework to encode STL tasks for our distributed control approach.
Let the scalar-valued function $b : \mathcal{D}\times \mathbb{R}_{+}\rightarrow \mathbb{R}$ with $\mathcal{D}\subset\mathbb{R}^{n\cdot N}$ being open and with corresponding level set $\mathcal{C}(t)$ defined as
\begin{equation}\label{eq:level set}
\mathcal{C}(t)=\{\vec{x}\in \mathcal{D} |\;  b(\vec{x},t)\geq0 \}.
\end{equation}
The following definitions are given
\begin{definition}\label{dLipschitz}
A scalar function $b: \mathcal{D} \times \mathbb{R}_+ \rightarrow \mathbb{R}$ is \textit{sd-differentiable} if it is differentiable over $\mathcal{D}$, piece-wise differentiable over $\mathbb{R}_+$ with discontinuities over a discrete set $S =\{\tau_0,\ldots \tau_q\}$ for some $q\geq 0$ and such that $\partial_{\vec{x}}b(\vec{x},t)$, $\partial_{t}b(\vec{x},t)$ are Lipschitz continuous over $\mathcal{D}\times [\tau_i,{\tau_{i+1}}), \forall i=0,\ldots q-1.$ Moreover, 1) the left limit $\lim_{t^- \rightarrow \tau_i} \mathcal{C}(t)\subseteq \mathcal{C}(\tau_i) \; \forall \tau_i \in S$, 2) $\tau_i=t^k$ for some $k\in \mathbb{N}_0$.
\end{definition}\par
In Def. \ref{dLipschitz}, Lipschitz continuity of the gradients $\partial_{\vec{x}}b(\vec{x},t)$ implies Lipschitz continuity of the Lie derivative $L_{\bar{f}}b(\vec{x},t)$ and $L_{\bar{g}}b(\vec{x},t)$, which is relevant for the proof of Lemma \ref{lemma local sample margin}. At the same time condition 1) intuitively imposes that $\mathcal{C}(t)$ is at least locally expanding at points of time discontinuity in $S$. Eventually condition 2) enforces times of discontinuities to correspond to sampling instants. Note that this assumption is not restrictive, but introduced nevertheless to avoid technical non-smooth analysis arguments in the proof of Lemma \ref{lemma local sample margin}.
\begin{definition}\label{forward invariance}
The set $\mathcal{C}(t)$ is \textit{forward invariant} over the time interval $[t^0,\tau]$ under $\vec{u}(t)\in \prod_i \mathcal{U}_i^{\delta t}$ if there exists a unique solution $\vec{x}(t)\in \prod_i\mathcal{X}_i$ to \eqref{eq:multi agent dynamics}, such that $\vec{x}(t^0)\in \mathcal{C}(t^0) \Rightarrow \vec{x}(t)\in \mathcal{C}(t) ,\; \forall t\in [t^0,\tau]$.
\end{definition}\par
Next, provide a definition of sdCBFs that are piece-wise differentiable over time. Note that continuous differentiability was assumed in \cite{roque2022corridor,breeden2021control}.
\begin{definition}\label{sdcbf definition} Consider a sampling interval $\delta t>0$,  system \eqref{eq:multi agent dynamics} and the sd-differentiable function $b: \mathcal{D} \times \mathbb{R}_+ \rightarrow \mathbb{R}$ where $\mathcal{C}(t)$ in \eqref{eq:level set} is compact and  such that $\mathcal{C}(t)\subset \mathcal{D}, \; \forall t \geq0$, over the open set $\mathcal{D}$. Then $b(\vec{x},t)$ is a sampled-data Control Barrier Function (sdCBF) if $\bar{g}(\vec{x})\partial_{\vec{x}} b(\vec{x},t) =\vec{0} \Leftrightarrow \partial_{\vec{x}} b(\vec{x},t) = \vec{0}$ and there exists a \textit{margin function} $\zeta : \mathbb{R}^n\times \mathbb{R}_+ \rightarrow \mathbb{R}$ defined as 
\begin{equation}\label{eq: local zeta}
\zeta(\vec{x},t) = \lambda b(\vec{x},t) + \partial_t b(\vec{x}, t) + \nu(\vec{x},t),
\end{equation}
where $\lambda \in \mathbb{R}_+$ and  $\nu(\vec{x},t)\leq \bar{\nu}(\vec{x},t)\; \forall (\vec{x},t) \in \mathcal{D}\times \mathbb{R}_+ $ with\footnote{For a scalar function $b(\vec{x},t):\mathbb{R}^{n\cdot N} \rightarrow \mathbb{R}$ the Lie derivative of $b$ along $\bar{f}$, $\bar{g}$, $f_i$ and $g_i$ is given as $L_{\bar{f}}b(\vec{x},t):= \partial_{\vec{x}}b(\vec{x},t)\bar{f}(\vec{x})$, $L_{\bar{g}}b(\vec{x},t):= \partial_{\vec{x}}b(\vec{x},t)\bar{g}(\vec{x})$, $L_{f_i}b(\vec{x},t):= \partial_{\vec{x}_i}b(\vec{x},t)f_i(\vec{x}_i)$, $L_{g_i}b(\vec{x},t):= \partial_{\vec{x}_i}b(\vec{x})g_i(\vec{x})$.}
\begin{equation}\label{eq:local nu}
\begin{gathered}
\bar{\nu}(\vec{x},t) = \min_{\bar{\vec{x}}\in \mathcal{R}(\vec{x},\delta t),\tau\in [t^k, t^{k+1}), \vec{u}\in \mathbb{U}}  \\
L_{\bar{f}} b(\vec{x}, t) + L_{\bar{g}}b(\vec{x},t)\vec{u}  - (L_{\bar{f}} b(\bar{\vec{x}}, \tau) + L_{\bar{g}}b(\bar{\vec{x}},\tau)\vec{u}) +  \\
\partial_t b(\vec{x}, t) -  \partial_t b(\bar{\vec{x}}, \tau) +
\lambda b(\vec{x}, t) -  \lambda b(\bar{\vec{x}}, \tau));
\end{gathered}
\end{equation}
such that there exist a p.w.c input $\vec{u}(t)\in \prod_i\mathcal{U}^{\delta t}_i$ as per \eqref{eq:piece-wise constant} satisfying the condition
\begin{equation}\label{eq: sdcbf condition}
\begin{aligned}
L_{\bar{f}} b(\vec{x}^k, t^k) + L_{\bar{g}}b(\vec{x}^k,t^k)\vec{u}^k \geq  - \zeta(\vec{x}^k,t^k)
\end{aligned}
\end{equation}
for all $\vec{x}^k\in \mathcal{D}$ and $t^k$ with $k\in \mathbb{N}_{0}$.
\end{definition}\par
The term $\bar{\nu}(\vec{x},t)$ intuitively represents the minimum local negative variation of the classical CBF constraint $L_{\bar{f}} b(\vec{x}, t) + L_{\bar{g}}b(\vec{x},t)\vec{u} + \partial_t b(\vec{x}, t) + \lambda b(\vec{x}, t) \geq 0$ at $(\vec{x},t)$, which defines the local effect of sampling over the constraint satisfaction. Note that $\bar{\nu}(\vec{x},t)$ exists and is bounded since $L_{\bar{f}}b(\vec{x},t),L_{\bar{g}}b(\vec{x},t), \lambda b(\vec{x},t),\partial_{t}b(\vec{x},t)$ are Lipschitz over the intervals $\mathcal{D}\times [t^k,t^{k+1})$. The following lemma inspired by \cite[Thm.2]{breeden2021control} justifies the definition of an sdCBF.
\begin{lemma}\label{lemma local sample margin}
Assume $b(\vec{x},t)$ is an sdCBF as per Def. \ref{sdcbf definition} for a sampling interval $\delta t>0$ and let $\vec{x}^0 = \vec{x}(t^0)\in \mathcal{C}(t^0)$. If \eqref{eq:multi agent dynamics} is subject to a p.w.c input trajectory $\vec{u}(t)\in \prod_i\mathcal{U}_i^{\delta t}$ such that condition \eqref{eq: sdcbf condition} is satisfied for all $\vec{x}^k\in \mathcal{D}$ and $t^k$ with $k\in \mathbb{N}_{0}$, then $\mathcal{C}(t)$ is forward invariant over $[t^0, \infty]$. 
\end{lemma}
\begin{proof}
Since \eqref{eq:multi agent dynamics} are Lipschitz and $\vec{u}(t)$ is p.w.c and bounded, then for any $\vec{x}(t^0)\in \mathcal{C}(t^0)\subset \mathcal{D} \subset \mathbb{X}$ the solution $\vec{x}(t)$ to the dynamics in \eqref{eq:multi agent dynamics} under $\vec{u}(t)$ is  absolutely continuous and unique for some interval $J=[t^0, T]\subset \mathbb{R}_+$ \cite[Thm. 54]{sontag2013mathematical}, where $T> t^0$. We next prove by contradiction that $T=\infty$ and that $\vec{x}(t) \in \mathcal{C}(t), \; \forall t\in [0,\infty)$. Namely, assume instead $\vec{x}(t)$ is defined up to a finite $T$ such that $T< t^1= t_0+\delta t$ and let $\vec{x}(t^0)\in \mathcal{C}(t^0)$. Since $b(\vec{x},t) : \mathcal{D}\times \mathbb{R}_+$ is an sdCBF, $\vec{u}(t)\in \prod_i\mathcal{U}_i^{\delta t}$ satisfies \eqref{eq: sdcbf condition} and $\vec{u}(t)=\vec{u}^0 \; \forall t\in [t^0,T]\in [t^0,t^1)$, then it holds the total time derivative $\dot{b}(\vec{x},t)$, is given by $\dot{b}(\vec{x}(t),t) = L_{\bar{f}} b(\vec{x}(t),t) + L_{\bar{g}}b(\vec{x}(t),t)\vec{u}^0 + \partial_t b(\vec{x}(t),t) =  L_{\bar{f}} b(\vec{x}(t),t) + L_{\bar{g}}b(\vec{x}(t),t)\vec{u}^0 + \partial_t b(\vec{x}(t),t) +
\Bigl(L_{\bar{f}} b(\vec{x}(t^0),t^0) + L_{\bar{g}}b(\vec{x}(t^0),t^0)\vec{u}^0 + \partial_t b(\vec{x}(t^0),t^0) + \lambda b(\vec{x}(t^0),t^0)\Bigr) - \Bigl(L_{\bar{f}} b(\vec{x}(t^0),t^0) + L_{\bar{g}}b(\vec{x}(t^0),t^0)\vec{u}^0 + \partial_t b(\vec{x}(t^0),t^0) +  \lambda b(\vec{x}(t^0),t^0)\Bigr)  + \lambda b(\vec{x}(t),t)  - \lambda b(\vec{x}(t),t)  \underset{\eqref{eq:local nu}}{\geq}  L_{\bar{f}} b( \vec{x}^0,t^0) + L_{\bar{g}}b(\vec{x}(t^0),t^0)\vec{u}^0 + \partial_t b(\vec{x}(t^0),t^0) +  \lambda b(\vec{x}(t^0),t^0) + \nu(\vec{x}(t^0),t^0) - \lambda b(\vec{x}(t),t) \underset{\eqref{eq: sdcbf condition}}{\geq} -\lambda b(\vec{x}(t),t)$.
By the Comparison Lemma \cite{hassan2002nonlinear} we have $\vec{x}(t^0)\in \mathcal{C}(t^0) \land \dot{b}(\vec{x}(t),t)\geq -\lambda b(\vec{x}(t),t) ,\, \forall t\in [t^0,T] \Rightarrow \vec{x}(t)\in \mathcal{C}(t), \, \forall t\in [t^0,T]$. Hence, the maximal solution $\vec{x}(t)$ evolves inside the compact set $\mathcal{C}(t)$. Since the maximal solution $\vec{x}(t)$ evolves inside a compact set and since the dynamics \eqref{eq:multi agent dynamics} is time-invariant, then it must be $T=t^1$ \cite[Prop C.3.6]{sontag2013mathematical}, contradicting the assumption $T < t^1$. Now let instead assume $t^1\leq T < t^2 = t_0+2\delta t$. Since $b(\vec{x},t)$ is sd-differentiable then it is known that either $t^1$ is a point of discontinuity for $b(\vec{x},t)$ or not. If $t^1$ is \textit{not} a point of discontinuity we have that $\lim_{t^-\rightarrow t^1}\mathcal{C}(t) = \mathcal{C}(t^1)$ and since $\vec{x}(t)$ is absolutely continuous we also have $\lim_{t^- \rightarrow t^1} \vec{x}(t) = \vec{x}(t^1)$. Thus, since $\vec{u}^0$ respects the sdCBF condition \eqref{eq: sdcbf condition} over $[t^0, t^1)$ we have $\lim_{t^-\rightarrow t^1}\vec{x}(t) =\vec{x}(t^1) \in \lim_{t^-\rightarrow t^1}\mathcal{C}(t)= \mathcal{C}(t^1)$. On the other hand, if $t^1$ is a point of discontinuity then it must be $\lim_{t^- \rightarrow t^1} \mathcal{C}(t) \subseteq \mathcal{C}(t^1)$ by Def. \ref{dLipschitz} for $b(\vec{x},t)$. Similarly to the previous arguments we have $\lim_{t^-\rightarrow t^1}\vec{x}(t) =\vec{x}(t^1) \in \lim_{t^-\rightarrow t^1}\mathcal{C}(t)\subseteq \mathcal{C}(t^1)$ such that we arrive at the conclusion that $\vec{x}(t)\in \mathcal{C}(t),\; \forall t\in [t_0,t_1]$. Morover, by the fact that $\vec{x}(t_1)\in \mathcal{C}(t^1)$ and since $\vec{u}^1$ respects \eqref{eq: sdcbf condition} over the interval $[t^1,t^2)$, we also  have $\vec{x}(t) \in \mathcal{C}(t), \; \forall t\in [t^0,T] \subset [t^0,t^2)$. By \cite[Prop C.3.6]{sontag2013mathematical}  we have again a contradiction since  $\mathcal{C}(t), \; \forall t\in [t^0,T]$, is compact and thus it must be (as in the previous case) $T=t^2$ and not $T< t^2$ as assumed. Since  the same arguments apply for any possible $t^0\leq T<t^k$, we conclude that $T=\infty$ and $\vec{x}(t)\in \mathcal{C}(t)$ for all $t\in [t^0,\infty]$.
\end{proof}
\begin{remark}
   From Def. \ref{sdcbf definition}, the discontinuity points $S=\{ \tau_i\}_i^q$ of $b(\vec{x},t)$ are limited to sampling instants $t^k$. Hence, partial derivatives are well defined within the sampling intervals $[t^k,t^{k+1})$. At the time of discontinuities, we consider the right partial derivatives.
\end{remark}
\subsection{sdCBF for STL tasks}
Given a general STL formula $\phi$ from fragment \eqref{eq:working fragment} such that the predicates $h_{i}$ and $h_{ij}$ respect the concavity assumption in Assmp. \ref{concavity},  then the authors in \cite{lindemann2018control} define a three steps procedure to construct a function $b^{\phi}(\vec{x},t)$ with associated set $\mathcal{C}^{\phi}(t)$ as per \eqref{eq:level set} such that forward invariance of $\mathcal{C}^{\phi}(t)$ implies $\vec{x}(t)\models \phi$ \cite[Sec. A]{lindemann2018control}. The resulting $b^{\phi}(\vec{x},t)$ is an sd-differentiable function as per Def. \ref{dLipschitz} assuming the predicate functions are differentiable and the time intervals $[a,b]$ for the temporal operators of fragment \eqref{eq:working fragment} are such that $a=t^{k_a},b=t^{k_b}$ for some $k_a,k_b\in \mathbb{N}_{0}$. Thus, for each task $\phi_i$ and $\phi_{ij}$,  sd-differentiable functions $b^\phi_{i}(\vec{x}_i,t)$, $b^\phi_{ij}(\vec{e}_{ij},t)$ can be constructed as per steps A,B,C in \cite[Sec. A]{lindemann2018control} such that $b^\phi_{i}(\vec{x}_i,t)$ and $b^\phi_{ij}(\vec{e}_{ij},t)$ are concave in their first argument. The next proposition relates the satisfaction of a task $\phi$ with the forward invariance of the level set $\mathcal{C}^{\phi}(t)$.
\begin{proposition}\label{stl satisfaction prop}
    Consider a task $\phi$ from \eqref{eq:working fragment} and let $b^{\phi}(\vec{x},t)$ be sd-differentiable and constructed according to \cite[Sec. A]{lindemann2018control} with corresponding level set $\mathcal{C}^{\phi}(t)$ as per \eqref{eq:level set}. If $\mathcal{C}^{\phi}(t)$ is forward invariant over $[t^0,\infty]$, then $\vec{x}(t)\models \phi$.
\end{proposition}
\begin{proof}
See \cite[Cor. 1 ]{larsTCNS}  and \cite[Thm. 1]{lindemann2018control}
\end{proof}
The following assumption is further considered.
\begin{assumption}\label{zero gradient assumption}
    Consider $\phi$ from \eqref{eq:working fragment}, the associated function $b^{\phi}(\vec{x},t)$ from Prop. \ref{stl satisfaction prop}. Furthermore,  let $\mathcal{B}^{\phi}(t) := \{ \vec{x}\in \mathcal{D} | \partial_{\vec{x}}b(\vec{x},t) = \vec{0} \}$ and 
    $\zeta^{\phi}(\vec{x},t) = \lambda b^{\phi}(\vec{x},t) + \partial_t b(\vec{x},t) + \nu(\vec{x},t)$  from \eqref{eq: local zeta}. Then $\lambda\in \mathbb{R}_+$ can be chosen such that $\zeta^{\phi}(\vec{x},t)>  \chi,\; \forall \vec{x}\in \mathcal{B}^{\phi}(t) $ for some $\chi > 0$.
\end{assumption}
\begin{remark}\label{gradient remark}
    From Def. \ref{sdcbf definition} we have $L_{\bar{g}}b^{\phi}(\vec{x},t)=\vec{0}\Leftrightarrow \partial_{\vec{x}}b^\phi(\vec{x},t)=\vec{0} \Rightarrow \vec{x}\in \mathcal{B}^{\phi}(t)$. Thus Assmp. \ref{zero gradient assumption} allows for the satisfaction of the sdCBF constraint \eqref{eq: sdcbf condition} when $\vec{x}(t)\in \mathcal{B}^{\phi}(t)$. Since by construction $b^{\phi}(\vec{x},t)$ is concave then $\mathcal{B}^{\phi}(t)$ corresponds to the set of maximisers of $b^{\phi}(\vec{x},t)$ at time $t$ such that $b^{\phi}(\vec{x},t) = b^{\phi}_{max}(t) = max_{\vec{x}\in \mathcal{D}}\,  b(\vec{x},t) \; \forall \vec{x}\in \mathcal{B}^{\phi}(t)$. Thus, Assmp. \ref{zero gradient assumption} implies choosing a sufficiently large $\lambda \in \mathbb{R}_+$ such that $\lambda b^{\phi}_{max}(t) + \partial_t b^{\phi}(\vec{x},t) + \nu^{\phi}(\vec{x},t)>  \chi \; \forall \vec{x}\in \mathcal{B}^{\phi}(t)\, \forall t\in \mathbb{R}_+$.  
\end{remark}

\section{Decentralized Control Solution}\label{control}
In this section, we define our control solution for Problem \ref{main problem}. The presentation is ordered as follows: In Sec. \ref{novel sol} we present the global decentralized control scheme applied to satisfy objective 1) in Problem \ref{main problem} assuming $\mathcal{G}_{\psi}\subseteq \mathcal{G}_{c}(t)$, for all $t$. An algorithmic implementation of our control scheme is given in Sec. \ref{token passing algorithm section} - \ref{control reduction sub section}. Communication maintenance is addressed in Corollary \ref{communicationc corollary}. For ease of presentation, we often adopt the shorthand notations $b_{ij}(t^k):= b_{ij}(\vec{e}_{ij}^k,t^k)$ and $b_i(t^k):= b_{i}(\vec{x}_{i}^k,t^k)$ when clear from the context and to reduce the notation.
\subsection{Decentralized Control for STL satisfaction}\label{novel sol}
Consider a given task $\phi_{ij}$ involving the relative state $\vec{e}_{ij}$ among agent $i$ and $j$, with associated sdCBF $b_{ij}^{\phi}(\vec{e}_{ij},t)$. By expanding \eqref{eq: sdcbf condition} it follows that 
\begin{equation}\label{eq:collaborative form lie derivative}
\begin{aligned}
L_{\bar{f}}b^{\phi}_{ij}(\vec{e}_{ij},t)+&L_{\bar{g}}b^{\phi}_{ij}(\vec{e}_{ij},t)\vec{u} = 
\\\hspace{1cm}& L_{f_i}b^{\phi}_{ij}(\vec{e}_{ij},t)+L_{g_i}b^{\phi}_{ij}(\vec{e}_{ij},t)\vec{u}_i + \\& \hspace{1cm}L_{f_j}b^{\phi}_{ij}(\vec{e}_{ij},t)+L_{g_j}b^{\phi}_{ij}(\vec{e}_{ij},t)\vec{u}_j,
\end{aligned}
\end{equation}
such that, for agent $i$, we define the \textit{worst impact} and \textit{best impact} over the satisfaction of $\phi_{ij}$, respectively as
\begin{subequations}\label{eq:impacts}
\begin{align}
{}^{i}\underbar{$\epsilon$}_{ij}^{\phi}(t^k)
    = \min_{\vec{u}_i\in \gamma^k_i\mathbb{U}_i} L_{f_i} b^{\phi}_{ij}(t^k) + L_{g_i} b^{\phi}_{ij}(t^k)\vec{u}_i \label{eq: follower worse impact}\\
    {}^{i}\bar{\epsilon}_{ij}^{\phi}(t^k) =\max_{\vec{u}_i\in \gamma^k_i\mathbb{U}_i} L_{f_i} b^{\phi}_{ij}(t^k) + L_{g_i} b^{\phi}_{ij}(t^k)\vec{u}_i, \label{eq: leader best impact}
   \end{align}
\end{subequations}
where $\gamma^k_{i}\in (0,1],\; \forall i\in \mathcal{V}_i, \forall k\in \mathbb{N}_0$, is denoted as the \textit{control reduction factor} which applies shrinks the control set $\mathbb{U}_i$ to $\gamma_i^k\mathbb{U}_i$. The purpose of $\gamma_i^k$ is to bound the value of the best and worse impacts as per \eqref{eq:impacts} to enforce the satisfaction of the sdCBF \eqref{eq: sdcbf condition} over $b_{ij}^{\phi}(\vec{e}_{ij},t)$ to satisfy task $\phi_{ij}$ as clarified in Section \ref{control reduction sub section}. \par
In addition to the best and worse impact terms in \eqref{eq:impacts},  the set of tokens  $T_i = \{ q_{ij} | \forall j \in \mathcal{N}_{\psi}(i) \}$ is introduced such that $q_{ij}\in \{l, f, u\}$  where $l$,$u$,$f$ stand for \textit{leader}, \textit{undefined} and \textit{follower} respectively and  
\begin{equation}\label{eq:token rules}
\begin{gathered}
q_{ij}=l \Leftrightarrow q_{ji}=f,\quad  q_{ij}=u \Leftrightarrow q_{ji}=u,\\ q_{ij}=l \Rightarrow q_{ik}\neq l \; \forall k\in \mathcal{N}_{\psi}(i)\setminus \{j\}.
\end{gathered}
\end{equation}
The purpose of the tokens $q_{ij}$ is to define an edge leader for each $(i,j)\in \mathcal{E}_{\psi}$. The rules in \eqref{eq:token rules} ensure that for each $(i,j) \in \mathcal{E}_{\psi}$, either $i$ or $j$ is a leader and that each $i$ is the leader of at most one edge. The undefined token $u$ only serves for initialization of Algorithm \ref{alg:leadership propagation}, from which token sets $T_i$ are constructed in a decentralized fashion (Sec. \ref{token passing algorithm section}). Moreover, let $\mathcal{L}_{\psi}(i):= \{j \in \mathcal{N}_{\psi}(i)| q_{ij} = f\}$ be the leader set for agent $i$, corresponding to all the agents $j$ for which $i$ is a follower of $(i,j)\in \mathcal{E}_{\psi}$. 
From these new terms, the main result of this work is presented next
\begin{theorem}\label{main control theorem}
Consider \eqref{eq:multi agent dynamics}, a task graph $\mathcal{G}_{\psi} \subseteq \mathcal{G}_c(t),\, \forall t\in \mathbb{R}_+$ such that tasks $\phi_{ij} \; \forall (i,j) \in \mathcal{E}_{\psi}$ and $\phi_i, \; \forall i\in \mathcal{V}$ are associated with corresponding $b_{ij}^{\phi}(\vec{e}_{ij},t)$, $b_{i}^{\phi}(\vec{x}_{i},t)$ constructed as per \cite[Sec. A]{lindemann2018control}. Furthermore, let $\zeta_{ij}^\phi(\vec{e}_{ij},t)$, $\zeta_{i}^\phi(\vec{x}_i,t)$ as per \eqref{eq: local zeta} and respecting Assmp. \ref{zero gradient assumption}. Let each agent $i$ be subject to the p.w.c input $\vec{u}_i(t)\in \mathcal{U}_{i}^{\delta t}$ such that  $\vec{u}_{i}^k$ minimizes 
\begin{subequations}
\label{eq:new decentralized optimization with tasks}
    \begin{align}
        & \hspace{1cm} \min_{\vec{u}_i,s_i,s_{ij}} \vec{u}_i^T\vec{u}_i + c_i \, s_{i} + \sum_{j \in \mathcal{L}_{\psi}(i)} c_{ij}\,s_{ij} \\
        \begin{split}
        &L_{f_i}b^{\phi}_{ij}(t^k) + L_{g_i}b^{\phi}_{ij}(t^k)\vec{u}_{i} \geq -\zeta^{\phi}_{ij}(t^k) -{}^{j}\underbar{$\epsilon$}^{\phi}_{ij}(t^k),\, q_{ij}=l\end{split}\label{eq: leader constraint}\\
        &L_{f_i}b^{\phi}_{ij}(t^k) + L_{g_i}b^{\phi}_{ij}(t^k)\vec{u}_{i} \geq -\frac{1}{2}\zeta^{\phi}_{ij}(t^k) - s_{ij}, \; \forall j\in \mathcal{L}_{\psi}(i) \label{eq:follower constraints}\\
        &L_{f_i}b^{\phi}_{i}(t^k) + L_{g_i}b^{\phi}_{i}(t^k)\vec{u}_{i} \geq -\zeta^{\phi}_{i}(t^k) - s_{i}, \label{eq:independednt slack constraint}\\
        &s_i \geq 0,\; s_{ij} \geq 0 \; \forall j\in \mathcal{L}_{\psi}(i),\; \vec{u}_i \in \gamma_i^k \mathbb{U}_i, \label{eq: slack variables}
    \end{align}
\end{subequations}
where $s_{i},s_{ij}$ are slack variables with corresponding positive cost coefficients $c_i,c_{ij}\in \mathbb{R}_+$, $\gamma_i^k\in (0,1] \;\forall i\in \mathcal{V},\, \forall k\in \mathbb{N}_0 $  and ${}^{j}\underbar{$\epsilon$}_{ij}^{\phi}(t^k)$ is the worse impact from the follower $j$ over task $\phi_{ij}$ as per \eqref{eq: follower worse impact}.
If \eqref{eq:new decentralized optimization with tasks} is feasible for each $t^k$ with $k\in \mathbb{N}_0$ then the solution $\vec{x}(t)$ to \eqref{eq:multi agent dynamics} under $\vec{u}(t)=[\vec{u}_i(t)^T]^T_{i\in \mathcal{V}}$ is such that $\vec{x}(t)\models \psi_{col}$. Moreover, if $s_i=0$ for all $t^k$, then  $\vec{x}(t)\models \psi_{col} \land \psi_{ind}$.
\end{theorem}
\begin{proof}
    We first prove $\vec{x}(t)\models \psi_{col}$. Consider agent $i$ with collaborative task $\phi_{ij}$ and sd-differentiable function $b_{ij}^{\phi}(\vec{e}_{ij},t)$ with $\zeta_{ij}^{\phi}(\vec{e}_{ij},t)$ as per \eqref{eq: local zeta} and $\mathcal{C}^{\phi}_{ij}(t)$ as per \eqref{eq:level set}. Further let $q_{ij}=l$ and $q_{ji}=f$ ($i$ is leader of $(i,j) \in \mathcal{E}_{\psi}$). Writing the left hand side of the sdCBF condition \eqref{eq: sdcbf condition} over $b_{ij}^{\phi}(\vec{e}_{ij},t)$ gives 
    $L_{f_i}b^{\phi}_{ij}(t^k) + L_{g_i}b^{\phi}_{ij}(t^k)\vec{u}_{i} + L_{f_j}b^{\phi}_{ij}(t^k) + L_{g_j}b^{\phi}_{ij}(t^k)\vec{u}_{j}  \underset{\eqref{eq: leader constraint}}{\geq} -\zeta_{ij}^{\phi}(t^k) -  {}^{j}\underbar{$\epsilon$}^{\phi}_{ij}(t^k) + L_{f_j}b^{\phi}_{ij}(t^k) + L_{g_j}b^{\phi}_{ij}(t^k)\vec{u}_{j} \underset{\eqref{eq: follower worse impact}}{\geq}  -\zeta_{ij}^{\phi}(t^k)$,
    which corresponds to the satisfaction of the sdCBF constraint \eqref{eq: sdcbf condition}. Note that at the singularity points for which $L_{g_i}b^{\phi}_{ij}(\vec{e}_{ij},t)=\vec{0}$ we have from the chain rule that $\partial_{\vec{x}_i}b^{\phi}_{ij}(\vec{e}_{ij},t) = -\partial_{\vec{x}_{j}}b^{\phi}_{ij}(\vec{e}_{ij},t)$ and thus by Def. \ref{sdcbf definition} we have $L_{g_i}b^{\phi}_{ij}(\vec{e}_{ij},t)=\vec{0} \Leftrightarrow L_{g_j}b^{\phi}_{ij}(\vec{e}_{ij},t)=\vec{0} \Rightarrow \partial_{\vec{x}} b_{ij}^{\phi}(\vec{e}_{ij},t) = \vec{0} \Rightarrow {}^{j}\underbar{$\epsilon$}^{\phi}_{ij}(t^k) = 0 $. Thus by Assmp. \ref{zero gradient assumption} on $\zeta_{ij}^{\phi}$ we have that \eqref{eq: leader constraint} is satisfied independently of $\vec{u}_i^k$ or $\vec{u}_j^k$ in this situation. Letting $\mathcal{C}^{\phi}_{ij}(t)$ be the level set of $b_{ij}^{\phi}$ as per Prop. \ref{eq:level set}, then from Lemma \ref{lemma local sample margin} we have $\vec{x}(t)\in \mathcal{C}^{\phi}_{ij}(t),\, \forall t\in \mathbb{R}_+$ and  thus $\vec{x}(t)\models \phi_{ij}$ by Prop. \ref{stl satisfaction prop}. Since the argument can be repeated for every edge $(i,j)\in \mathcal{E}_{\psi}$ then $\vec{x}(t)\models \psi_{col}$. Turning to the satisfaction of $\psi_{ind}$, consider independent tasks $\phi_i$ with associated $b_{i}^{\phi}(\vec{x}_i,t)$, $\zeta^{\phi}_{i}(\vec{x}_i,t)$ and level set $\mathcal{C}^{\phi}_{i}(t)$ as per \eqref{eq:level set}. It is clear that if $s_{i}=0 \; \forall i\in \mathcal{V}$ and for all $t^k$ then constraint \eqref{eq:independednt slack constraint} corresponds to the sdCBF condition 
    \eqref{eq: sdcbf condition} and thus $\vec{x}_i(t)\in \mathcal{C}^{\phi}_i(t) \, \forall t\in \mathbb{R}_+$. From Prop. \ref{stl satisfaction prop} we then have $\vec{x}(t)\models \phi_{i}, \, \forall i\in \mathcal{V}\Rightarrow \vec{x}(t)\models \psi_{ind}$.
\end{proof}
\par 
It is highlighted that constraint \eqref{eq: leader constraint} represents the sdCBF constraint \eqref{eq: sdcbf condition} over $b_{ij}^{\phi}(\vec{e}_{ij},t)$, as per \eqref{eq:collaborative form lie derivative}, where the follower $j$ is assumed to select the input $\vec{u}_j \in \gamma_{j}^k\mathbb{U}_j$ associated with the worst impact as per \eqref{eq: follower worse impact}. The control law \eqref{eq:new decentralized optimization with tasks} differs from previous controllers \cite{panagou2015distributed,larsTCNS,nonsmoothBarriers} as follows: \textbf{1)} Unlike \cite{larsTCNS,nonsmoothBarriers}, agent $i$ requires only local knowledge of state $\vec{x}_i$ and $\vec{e}_{ij}, \forall j\in \mathcal{N}_{\psi}(i)$. \textbf{2)} In contrast to \cite{larsTCNS,panagou2015distributed}, and akin to \cite{nonsmoothBarriers}, each agent is subject to $|\mathcal{N}_{\psi}(i)|$ constraints from collaborative tasks $\phi_{ij}$ plus one constraint from independent task $\phi_i$. However, unlike \cite{nonsmoothBarriers}, all constraints \eqref{eq:follower constraints}-\eqref{eq:independednt slack constraint} are satisfied only up to a slack variable, while only \eqref{eq: leader constraint} is strictly satisfied, corresponding to the unique edge $(i,j)\in\mathcal{E}_{\psi}$ for which $i$ is the leader ($q_{ij}=l$). \textbf{3)} Unlike \cite{larsTCNS,panagou2015distributed,nonsmoothBarriers,distributedSecondOrder}, the control law \eqref{eq:new decentralized optimization with tasks} is implemented in a sampled-data fashion which is suitable for embedded-systems implementation. \par
In Theorem \ref{main control theorem} it is assumed that $\mathcal{G}_\psi \subseteq \mathcal{G}_c(t)$. Nevertheless, a communication task can be explicitly given over each edge by imposing $\varphi^{com}_{ij}= G_{[0,\infty]}h_{ij}^{com}$ where $h_{ij}^{com} = \|\vec{p}_{ij}\|^2 -r_c^2$, which respects Assumption \ref{concavity}.
\begin{corollary}\label{communicationc corollary}
    Consider \eqref{eq:multi agent dynamics} and task graph $\mathcal{G}_{\psi} \subseteq \mathcal{G}_c(t^0)$ such that the conjunction of tasks $\phi_{ij}$ is such that a communication task $\varphi^{com}_{ij}= G_{[0,\infty]}h_{ij}^{com}$ with $h_{ij}^{com} = \|\vec{p}_{ij}\|^2 -r_c^2$ is included in $\phi_{ij}$ as per \eqref{eq:multi agent spec}. If each agent $i$ is subject to a p.w.c input $\vec{u}_i(t)\in \mathcal{U}_{i}^{\delta t}$ as per Thm. \ref{main control theorem}, then $\vec{x}(t) \models \psi_{col} \Rightarrow \vec{x}(t) \models \varphi^{com}_{ij},\;  \forall (i,j) \in \mathcal{E}_{\psi}$ and $\mathcal{G}_\psi \subseteq \mathcal{G}_c(t),\; \forall t\in \mathbb{R}_+$.
\end{corollary}
The subsequent sections further elucidate how the token sets $T_i$ are defined and how ${}^j\underbar{$\epsilon$}_{ij}^{\phi}(t^k)$ and $\zeta_{ij}^{\phi}(t^k)$ are computed online at each $t^k$. 
\subsection{Token passing algorithm}\label{token passing algorithm section}
To implement the decentralized control law in \eqref{eq:new decentralized optimization with tasks}, edge leaders should be defined as they are responsible for the satisfaction of the collaborative tasks $\phi_{ij}$ by satisfying constraint \eqref{eq: leader constraint}. Since $\mathcal{G}_{\psi}$ is acyclic, we define the token sets $T_i$ from the leaf nodes to the root node of $\mathcal{G}_{\psi}$ in a decentralized way. The process starts by initialising the sets $T_i$ with all undefined tokens $q_{ij}=u$ (line \ref{lst:line:initialise}). Then, the leaf nodes in $\mathcal{F}_{\psi} = \{ i \in \mathcal{V} |\, |\mathcal{N}_{\psi}(i)|=1 \}$ immediately set their unique token $q_{ij}$ to $q_{ij}=l$ (line \ref{lst:line:leaf start}). At the same time, the non-leaf nodes $i\in \mathcal{V}\setminus \mathcal{F}_{\psi}$ continuously require tokens $q_{ji}$ from their neighbours $j\in \mathcal{N}_{\psi}(i)$ and update $q_{ij}$ from $q_{ij}=u$ to $q_{ij}=f$ as soon as a token  $q_{ji}=l$ is detected (lines \ref{lst:line:checking start}-\ref{lst:line:checking end}). When agent $i\in \mathcal{V}\setminus \mathcal{F}_{\psi}$ only has one undefined token left $q_{ij}=u$, it sets  $q_{ij} = l$ (line \ref{lst:line:final resolution}). Thus the sets $T_i$ are progressively defined from leaf nodes to the root node in at most $\lceil \frac{\rho}{2} \rceil$ rounds where $\rho$ is the diameter of $\mathcal{G}_{\psi}$ and such that rules \eqref{eq:token rules} are satisfied. Alg. \ref{alg:leadership propagation} terminates with either one node that is follower of all its collaborative tasks ($q_{ij}=f, \; \forall j\in \mathcal{N}_{\psi}(i)$), or with two nodes $i$,$j$ with a single undefined token $q_{ij}=q_{ji}=u$. In this case, some simple priority resolution can be applied to set $q_{ij}=f$ and $q_{ji}=l$, or \textit{vice-versa}.
\begin{algorithm}[h]
\caption{Leadership Token Passing Algorithm (Agent $i$)}\label{alg:leadership propagation}
\begin{algorithmic}[1]
\STATE Initialise $T_i = \{q_{ij}= u\; \forall j\in \mathcal{N}_{\psi}(i) \}$ \label{lst:line:initialise}
\IF{$i \in \mathcal{F}_\psi$} \label{lst:line:leaf start}
    \STATE Set the unique $q_{ij}=l$  
\ELSE
    \WHILE{ There is more than one $q_{ij}=u$ } \label{lst:line:checking start}
        \FOR{$j \in \mathcal{N}_{\psi}(i)$}
            \STATE Require $q_{ji}$ from $j$
            \IF{$q_{ji}= l$ and $q_{ij} = u$}
                \STATE $q_{ij} = f$
            \ENDIF
        \ENDFOR
    \ENDWHILE \label{lst:line:checking end}
    \STATE Set last token $q_{ij}=l$ if any \label{lst:line:final resolution}
\ENDIF
\end{algorithmic}
\end{algorithm}
\begin{algorithm}[b]
\caption{Control Reduction Computation (Agent $i$)}\label{alg:control reduction alg}
\begin{algorithmic}[1]
\STATE Compute ${}^{i}\Upsilon^{\phi}_{ij}(t^k)$ for each $j\in \mathcal{N}_{\psi}(i)$ and $\nu_{i}^{\phi}(t^k)$ \label{lst:line:entry point}
\IF{$i$ in $ \mathcal{F}_{\psi}$}\label{lst:line:leaf node gamma 1}
    \STATE Set $\gamma^k_i=1$
    \STATE Compute $\bar{\vec{u}}^k_{i}$ \eqref{eq:controls} and best impact  $^{i}\bar{\epsilon}_{ij}^{\phi}(t^k)$ \eqref{eq: leader best impact}.
\ELSE
    \WHILE{Any $j \in \mathcal{L}_{\phi}(i)$ has not yet fixed $\gamma^k_j$}
      \FOR{$j\in \mathcal{L}_{\phi}(i)$} 
            \IF{$j$ has fixed $\gamma^k_j$}  \label{lst:line:checking leaders availability}
                \STATE Receive $^{j}\bar{\epsilon}_{ij}^{\phi}(t^k)$ \eqref{eq: leader best impact}
                \STATE Receive ${}^{j}\Upsilon_{ij}^{\phi}(t^k)$  \eqref{eq: upsilon}
                \STATE Compute $\nu_{ij}^{\phi}(t^k)$ \eqref{eq: nu by sum} and $\zeta_{ij}^{\phi}(t^k)$  \eqref{eq:local nu}
                \STATE Compute $\underline{\vec{u}}^k_{i}$ \eqref{eq:controls} and $\tilde{\gamma}^k_{ij}$ \eqref{eq:tilde gamma}
            \ENDIF  \label{lst:line:checking leaders availability 2 }
        \ENDFOR
    \ENDWHILE
    \STATE Fix $\gamma^k_i = \min_{j\in \mathcal{L}_{\phi}(i) } \tilde{\gamma}_{ij}$ \label{lst:line:end}
\ENDIF
\end{algorithmic}
\end{algorithm}
\subsection{Margin function computation}\label{margin computation}
In \eqref{eq:new decentralized optimization with tasks} the margin functions $\zeta_{ij}^{\phi}(\vec{e}_{ij},t)=\lambda b^\phi(\vec{e}^k_{ij},t^k) + \partial_tb^{\phi}_{ij}(\vec{e}^k_{ij},t^k) + \nu_{ij}^{\phi}(\vec{e}_{ij}^k,t^k)$ for all $j\in \mathcal{N}_{\psi}(i)$ and $\zeta_{i}^{\phi}(\vec{x}_{i},t)=\lambda b_i^\phi(\vec{x}^k_{i},t^k) + \partial_tb^{\phi}_{i}(\vec{x}^k_{i},t^k) + \nu_{i}^{\phi}(\vec{x}_{i}^k,t^k)$ as per \eqref{eq: local zeta} for some (possibly different) $\lambda \in \mathbb{R}_+$ must be computed at each $t^k$. Particularly, $\nu_{i}^{\phi}(\vec{x}_{i}^k,t^k)$ and $\nu_{ij}^{\phi}(\vec{e}_{ij}^k,t^k)$ must satisfy $\nu_{i}^{\phi}(\vec{x}_{i}^k,t^k) \leq \bar{\nu}_{i}^{\phi}(\vec{x}_{i}^k,t^k) $ and $\nu_{ij}^{\phi}(\vec{e}_{ij}^k,t^k) \leq \bar{\nu}_{ij}^{\phi}(\vec{e}_{ij}^k,t^k)$ with  $\bar{\nu}_{ij}^{\phi}(\vec{e}_{ij}^k,t^k)$ and $\bar{\nu}_{i}^{\phi}(\vec{x}_{i}^k,t^k)$ being according to \eqref{eq:local nu}. To reduce conservativism, it is beneficial to compute values $\nu_{ij}^{\phi}(\vec{e}_{ij}^k,t^k)$, $\nu_{i}^{\phi}(\vec{x}_{i}^k,t^k)$ as close as possible to their optimal values $\bar{\nu}_{ij}^{\phi}(\vec{e}_{ij}^k,t^k)$, $\bar{\nu}_{i}^{\phi}(\vec{x}_{i}^k,t^k)$, recalling that this computation is accomplished online. Regarding $\nu_{i}^{\phi}(\vec{x}_{i}^k,t^k)$, it is possible to directly compute $\nu_{i}^{\phi}(\vec{x}_{i}^k,t^k) = \bar{\nu}_{i}^{\phi}(\vec{x}_{i}^k,t^k)$ by solving for \eqref{eq:local nu} online at each time $t^k$ since each agent $i$ is aware of $\vec{x}_i^k$ and can compute $\mathcal{R}_i(\vec{x}^k_i,\delta t)$, where $\mathcal{R}_i$ is the reachable set for agent $i$. On the other hand, we have that computing $\bar{\nu}_{ij}^{\phi}(\vec{e}_{ij}^k,t^k)$ online requires solving the following minimization program from \eqref{eq:local nu} :
$$
\begin{aligned}
&\bar{\nu}_{ij}^{\phi}(\vec{e}^k_{ij},t^k) = \min \bigl\{ L_{f_i} b_{ij}^\phi(\vec{e}^k_{ij},t^k) - L_{f_i} b_{ij}^\phi(\bar{\vec{e}}_{ij},\tau) +\\&(L_{g_i}b_{ij}^\phi(\vec{e}^k_{ij},t^k)- L_{g_i}b_{ij}^\phi(\bar{\vec{e}}_{ij},\tau))\vec{u}_i +
L_{f_j} b_{ij}^\phi(\vec{e}^k_{ij},t^k) - \\
&L_{f_j} b_{ij}^\phi(\bar{\vec{e}}_{ij},\tau) + 
(L_{g_j}b_{ij}^\phi(\vec{e}^k_{ij},t^k)- L_{g_j}b_{ij}^\phi(\bar{\vec{e}}_{ij},\tau))\vec{u}_j+\\
&\partial_t b_{ij}^\phi(\vec{e}^k_{ij},t^k) -  \partial_t b_{ij}^\phi(\bar{\vec{e}}_{ij},\tau)
 + \lambda b_{ij}^{\phi}(\vec{e}^k_{ij},t^k)-\lambda b_{ij}^\phi(\bar{\vec{e}}_{ij},\tau) \bigr\} \\
 &s.t : \; \bar{\vec{e}}_{ij}\in \mathcal{R}_{ij}(\vec{e}_{ij},\delta t),\tau \in [t^k,t^{k+1}), \vec{u}_{i}\in \mathbb{U}_i,\vec{u}_j\in \mathbb{U}_j;
\end{aligned}
$$
where $\mathcal{R}_{ij}(\vec{e}^k_{ij},\delta t) = \mathcal{R}_{i}(\vec{x}^k_{i},\delta t) \ominus \mathcal{R}_{j}(\vec{x}^k_{j},\delta t)$. Indeed, $\mathcal{R}_{ij}(\vec{e}^k_{ij},\delta t)$ represents the reachable set of the relative state between $i$ and $j$ within a sampling interval $\delta t$. Knowledge of both the dynamics of $i$ and $j$ as well as $\mathcal{R}_i(\vec{x}_i,\delta t)$, $\mathcal{R}_j(\vec{x}_j,\delta t)$ is required to compute $\bar{\nu}_{ij}^{\phi}(\vec{e}^k_{ij},t^k)$. Since we assume that agents do not share detailed knowledge of each other's dynamics, we only let $\mathcal{R}_i(\vec{x}_i^k,\delta t)$ and $\mathcal{R}_j(\vec{x}_j^k,\delta t)$ to be shared through communication among $i$ and $j$ at each sampling time $t^k$. Hence we define $\nu_{ij}^{\phi}(\vec{e}^k_{ij},t^k)$ as
\begin{equation}\label{eq: nu by sum}
\begin{aligned}
& \nu_{ij}^{\phi}(\vec{e}^k_{ij},t^k) ={}^{i}\Upsilon^\phi_{ij}(t^k)  +  {}^{j}\Upsilon^\phi_{ij}(t^k)
\end{aligned}
\end{equation}
where
\begin{subequations}\label{eq: upsilon}
\begin{align}
\begin{aligned}
&^{i}\Upsilon^{\phi}_{ij}(t^k) :=  \min_{\vec{u}_i, \bar{\vec{e}}_{ij},\tau} \bigl\{L_{f_i} b_{ij}^\phi(\vec{e}^k_{ij},t^k) -L_{f_i} b_{ij}^\phi(\bar{\vec{e}}_{ij},\tau) + \\
&(L_{g_i}b_{ij}^\phi(\vec{e}^k_{ij},t^k)-L_{g_i}b_{ij}^\phi(\bar{\vec{e}}_{ij},\tau))\vec{u}_i +\bigl(\lambda b_{ij}^{\phi}(\vec{e}^k_{ij},t^k)- \\
&\lambda b_{ij}^\phi(\bar{\vec{e}}_{ij},\tau) + \partial_t b_{ij}^\phi(\vec{e}^k_{ij},t^k) -  \partial_t b_{ij}^\phi(\bar{\vec{e}}_{ij},\tau) \bigr)/2\bigl\} \\
&s.t :\;   \bar{\vec{e}}_{ij}\in \mathcal{R}_{ij}(\vec{e}^k_{ij},\delta t), \, \tau\in [t^k, t^{k+1}),\, \vec{u}_i\in \mathbb{U}_i;
\end{aligned}\\
\begin{aligned}
&^{j}\Upsilon^\phi_{ij}(t^k) := \min_{\vec{u}_j, \bar{\vec{e}}_{ij},\tau} \bigl\{
L_{f_j} b_{ij}^\phi(\vec{e}^k_{ij},t^k) -L_{f_j} b_{ij}^\phi(\bar{\vec{e}}_{ij},\tau) +\\ 
&(L_{g_j}b_{ij}^\phi(\vec{e}^k_{ij},t^k)-L_{g_j}b_{ij}^\phi(\bar{\vec{e}}_{ij},\tau))\vec{u}_j +\bigl(\lambda b_{ij}^{\phi}(\vec{e}^k_{ij},t^k)-\\ 
&\lambda b_{ij}^\phi(\bar{\vec{e}}_{ij},\tau) + \partial_t b_{ij}^\phi(\vec{e}^k_{ij},t^k) -  \partial_t b_{ij}^\phi(\bar{\vec{e}}_{ij},\tau) \bigr)/2\bigl\}\\
&st: \; \bar{\vec{e}}_{ij}\in \mathcal{R}_{ij}(\vec{e}^k_{ij},\delta t),\,   \tau\in [t^k, t^{k+1}),\,  \vec{u}_j\in \mathbb{U}_j. 
\end{aligned}
\end{align}
\end{subequations}
and we note that $\nu_{ij}^{\phi}(\vec{e}^k_{ij},t^k) \leq \bar{\nu}_{ij}^{\phi}(\vec{e}^k_{ij},t^k)$\footnote{Summing the objectives in \eqref{eq: upsilon}, the original objective in \eqref{eq:local nu} is found. The inequality is then justified since the minimum of the sum is greater than the sum of the minimums}. Thus, for each task $\phi_{ij}$, agent $i$ computes $^{i}\Upsilon_{ij}^k$ while it receives the value of $^{j}\Upsilon_{ij}^k$ at each $t^k$ from all $j\in \mathcal{N}_{\psi}(i)$. Eventually $\nu_{ij}^{\phi}(\vec{e}^k_{ij},t^k)$ are computed by direct summation as per \eqref{eq: nu by sum}
\begin{remark}\label{remark on non convexity}
    Computing $\nu_{i}^{\phi}(\vec{x}^k_i,t^k)$ and $\nu_{ij}^{\phi}(\vec{e}^k_{ij},t^k)$ online, requires solving a nonlinear non-convex program. Hence, each agent solves for $^{j}\Upsilon^\phi_{ij}(t^k), \; \forall j\in \mathcal{N}_{\psi}(i)$ and $\nu_{i}^{\phi}(\vec{x}_i^k,t^k)$ by selecting a batch of initial conditions starting from which the parallel optimization is undertaken. The minimum solution is then taken as the actual solution.
\end{remark}
\subsection{Control reduction factor computation}\label{control reduction sub section}
At last the control reduction factors $\gamma_i^k,,\forall i\in \mathcal{V}$ and the worse impacts ${}^{j}\underbar{$\epsilon$}^{\phi}_{ij}(t^k)$ as per \eqref{eq: follower worse impact} must be computed to implement \eqref{eq:new decentralized optimization with tasks}. Initially, consider a single task $\phi_{ij}$ with corresponding sdCBF $b_{ij}^\phi(\vec{e}_{ij},t)$ and margin function $\zeta_{ij}^{\phi}(\vec{e}_{ij},t)$ as per \eqref{eq: local zeta}. Furthermore, let $i$ be the leader of the edge $(i,j)\in \mathcal{E}_{\psi}$ and assume at time $t^k$ that the control reduction factor $\gamma^k_i$ for the edge leader $i$ is fixed.  We next derive an analytical approach to compute $\gamma_j^k$ for the follower agent $j$ such that no matter how the follower $j$ chooses its input $\vec{u}_j \in \gamma_{j}^k \mathbb{U}_j$ at time $t^k$, there exist always $\vec{u}_i^k \in \gamma_i^k \mathbb{U}_i$ from the leader such that \eqref{eq: leader constraint} is satisfied. Let \begin{equation}
\begin{aligned}
V_{ij}^{\phi}(t^k) = \max_{\vec{u}_i\in \gamma^k_i\mathbb{U}_i} \min_{\vec{u}_j\in \gamma^k_j\mathbb{U}_j}
 L_{f_i} b_{ij}^\phi(t^k) + L_{g_i}b_{ij}^\phi(t^k)\vec{u}_i +\\ 
L_{f_j} b_{ij}^\phi(t^k) + L_{g_j}b_{ij}^\phi(t^k)\vec{u}_j +  \zeta_{ij}^\phi(t^k),
\end{aligned}
\end{equation}
where $V^{\phi}_{ij}(t^k)$ intuitively represents the value function of a two-player game over the sdCBF condition \eqref{eq: sdcbf condition} for task $\phi_{ij}$ where $i$ is the maximizing player (the leader) and $j$ is the minimizing player (the follower). If $\gamma_j^k$ is designed such that $V_{ij}^{\phi}(t^k)\geq 0$ at each time $t^k$, then there exists $\vec{u}_i\in \gamma_i^k\mathbb{U}_i$ from the leader such that constraint \eqref{eq: leader constraint} is satisfied leading to the satisfaction of $\phi_{ij}$. The following proposition is needed for the next derivations 
\begin{proposition}\label{eq: homogeneity}
Consider the minimization $\min_{x\in \mathcal{X}} f(x)$ 
where $\mathcal{X}$ is convex and $f(\alpha x)=\alpha f(x)$ with $\alpha >0$. If $x^*\in \text{argmax}_{x\in \mathcal{X}}(f(x))$ , then $\alpha x^*\in  \text{argmax}_{x_\alpha \in \alpha \mathcal{X}}(f(x))$.  
\end{proposition}
\begin{proof}
    By definition $x^*\in \text{argmax}_{x\in \mathcal{X}}(f(x))\Leftrightarrow f(x^*)\geq f(x)\; \forall x\in \mathcal{X}$. Exploiting the properties of $f(x)$ we have $f(\alpha x^*)=\alpha f(x^*)\geq \alpha f(x)=f(\alpha x)$ for all $x\in \mathcal{X}$. Since $\alpha \mathcal{X}:=\{\alpha x| \;\forall x\in \mathcal{X}\}$ we conclude  $f(\alpha x^*)\geq f(x_\alpha)\; \forall x_\alpha \in \alpha \mathcal{X} \Rightarrow \alpha x^*\in  \text{argmax}_{x_\alpha \in \alpha \mathcal{X}}(f(x))$.
\end{proof}\par
Since $\vec{u}_i$ and $\vec{u}_j$ enter linearly in $V^{\phi}_{ij}(t^k)$, we have from Prop. \ref{eq: homogeneity} and \eqref{eq: leader best impact}-\eqref{eq: follower worse impact} that 
\begin{equation}\label{eq:value function}
\begin{aligned}
V_{ij}^{\phi}(t^k)= L_{f_i} b_{ij}^\phi(t^k) + L_{g_i}b_{ij}^\phi(t^k)\bar{\vec{u}}^k_i\gamma^k_i + L_{f_j} b_{ij}^\phi(t^k) +  \\
L_{g_j}b_{ij}^\phi(t^k)\underline{\vec{u}}^k_j\gamma^k_j + \zeta_{ij}^\phi(t^k) = {}^{i}\bar{\epsilon}_{ij}^{\phi}(t^k) +  {}^{j}\underbar{$\epsilon$}_{ij}^{\phi}(t^k) + \zeta_{ij}^\phi(t^k)
\end{aligned}
\end{equation}
where
\begin{equation}\label{eq:controls}
\hspace{-0.2cm}\underline{\vec{u}}^k_{j} = \underset{\vec{u}_j\in \mathbb{U}_j}{\text{argmin}} L_{g_j}b^{\phi}_{ij}(t^k)\vec{u}_j\,;
\bar{\vec{u}}^k_{i} = \underset{\vec{u}_i\in \mathbb{U}_i}{\text{argmax}}L_{g_i}b^{\phi}_{ij}(t^k)\vec{u}_i.
\end{equation}
Hence, after agent $i$ and $j$ independently compute $\bar{\vec{u}}^k_{i}$, $\underline{\vec{u}}^k_{j}$ by solving \eqref{eq:controls}, an analytical expression of $V_{ij}^{\phi}(t^k)$ as a function of $\gamma_i^k$ and $\gamma_j^k$ is derived from \eqref{eq:value function}. The next assumption is considered.
\begin{assumption}\label{gamma positivity}
Consider task $\phi_{ij}$ such that $i$ is the leader of $(i,j)\in \mathcal{E}_{\psi}$. Then it holds ${}^{i}\bar{\epsilon}^{\phi}_{ij}(t^k)  + \zeta_{ij}^{\phi}(t^k) + L_{f_i}b_{ij}^{\phi}(t^k) \geq 0$ for all $t^k$ for system \eqref{eq:multi agent dynamics} under the control law \eqref{eq:new decentralized optimization with tasks}, where ${}^{i}\bar{\epsilon}_{ij}^\phi(t^k)$ is best impact $i$ as per \eqref{eq: leader best impact}.
\end{assumption}\par
Assumption \ref{gamma positivity} is introduced to ensure that the leader agent $i$ is always able to enforce $V_{ij}^{\phi}(t^k) \geq 0$ (and thus satisfy the sdCBF condition \eqref{eq: sdcbf condition}) if the follower agent $j$ chooses $\gamma_j^k=0$ at time $t^k$  (and thus $\vec{u}_{j}^k=\vec{0}$). At this point, once the follower agent $j$ computes $\zeta_{ij}^{\phi}(t^k)$ (Sec \ref{margin computation}) and the best impact ${}^{i}\bar{\epsilon}_{ij}^\phi(t^k)$ is sent from the leader $i$ to $j$ as per \eqref{eq: leader best impact}, then imposing $V_{ij}^{\phi}(t^k) \geq 0$ implies $\gamma^k_j \leq \frac{-({}^{i}\bar{\epsilon}^{\phi}_{ij}(t^k)  + \zeta_{ij}^{\phi}(t^k) + L_{f_i}b_{ij}^{\phi}(t^k)) }{L_{g_j}b_{ij}^{\phi}(t^k)\underline{\vec{u}}^k_{j}}$ if $L_{g_j}b_{ij}^{\phi}(t^k)\underline{\vec{u}}^k_{j}< 0$ or 
$\gamma^k_j \geq \frac{-({}^{i}\bar{\epsilon}^{\phi}_{ij}(t^k)  + \zeta_{ij}^{\phi}(t^k) + L_{f_i}b_{ij}^{\phi}(t^k)) }{ L_{g_j}b_{ij}^{\phi}(t^k)\underline{\vec{u}}^k_{j}}$ if $L_{g_j}b_{ij}^{\phi}(t^k)\underline{\vec{u}}^k_{j}\geq0$.
By Assmp. \ref{gamma positivity} the numerator of the previous expressions is non-positive so any $\gamma^k_j \in (0,1]$ would satisfy the latter conditioning and $\gamma_j^k=1$ is set arbitrarily in this case. On the other hand, the former condition is the only binding one. Since $j$ is possibly follower for a number of tasks equal to $|\mathcal{L}_{\psi}(i)|$ , we have that $\gamma_j^k$ is computed as $\gamma^k_{j} = \min_{i\in \mathcal{L}_{\psi}(j)}\{ \tilde{\gamma}^k_{ij} \}$ where
\begin{equation}\label{eq:tilde gamma}
    \tilde{\gamma}^k_{ij} = \begin{cases}
    1 \;\quad  \text{if}\; \quad  L_{g_j}b_{ij}^{\phi}(t^k)\underline{\vec{u}}^k_j \geq 0 \\
    \min\{1, \frac{-({}^{i}\bar{\epsilon}^{\phi}_{ij}(t^k)  + \zeta_{ij}^{\phi}(t^k) + L_{f_i}b_{ij}^{\phi}(t^k)) }{L_{g_j}b_{ij}^{\phi}(t^k)\underline{\vec{u}}^k_{j}}\Bigl\}\;  \text{else},
    \end{cases}
\end{equation}
$\forall i\in \mathcal{L}_{\psi}(j)$.  For the special case $L_{g_j}b_{ij}^{\phi}(t^k) = \vec{0}$, it is known from Assmp. \ref{zero gradient assumption} and Remark \ref{gradient remark} that  $L_{g_j}b_{ij}^{\phi}(t^k) = \vec{0} \Leftrightarrow L_{g_i}b_{ij}^{\phi}(t^k) = \vec{0} \Rightarrow \partial_{\vec{x}}b_{ij}^\phi(t^k) = \vec{0}$ and in this case $V_{ij}^{\phi} = \zeta_{ij}^{\phi}(t^k) \geq \chi$ for some $\chi > 0$. Hence $\tilde{\gamma}^k_{ij}$ can be chosen arbitrarily and we take the least conservative solution $\tilde{\gamma}^k_{ij}=1$ in \eqref{eq:tilde gamma}. After $\gamma^k_j$ is computed, the worse impact $^{j}\underbar{$\epsilon$}^\phi_{ij}(t^k)$ in \eqref{eq: follower worse impact} is also computed as $^{j}\underbar{$\epsilon$}^\phi_{ij}(t^k):= L_{f_j} b_{ij}^\phi(t^k) + L_{g_i}b_{ij}^\phi(t^k)\underline{\vec{u}}^k_j\gamma^k_j$ and sent back to each leader $i \in \mathcal{L}_{\psi}(j)$ so that constraint \eqref{eq: leader constraint} can be computed for each leader.\par
Given that the leaders $i\in \mathcal{L}_{\psi}(j)$ have fixed their values of $\gamma_i^k$, then we have shown how $\gamma_j^k$ such that a solution to \eqref{eq: leader constraint} exists for each leader. Next, the decentralized algorithm Alg. \ref{alg:control reduction alg} provides a way for all agents $i\in \mathcal{V}$ to fix their value of $\gamma_i^k$ sequentially as each $t^k$.  Namely, Alg. \ref{alg:control reduction alg} start with each agent first computing ${}^{i}\Upsilon^\phi_{ij}(t_k)$ as per \eqref{eq: upsilon} and $\nu_i^\phi(t^k)$ as per Sec. \ref{margin computation} (line \ref{lst:line:entry point}). Then, if agent $i$ is a leaf node in $\mathcal{F}_{\psi}$, then $\gamma_i^k$ is set to $\gamma_i^k=1$ and the best impact $^{i}\bar{\epsilon}^\phi_{ij}(t^k)$ is computed for the only task $\phi_{ij}$ as per \eqref{eq: leader best impact}. Selecting $\gamma_i^k=1$ for the leaf nodes is justified since leaf nodes do not have any leader. On the other hand, if $i\not\in \mathcal{F}_{\psi}$, then $i$ continuously checks its leader neighbours in $j\in \mathcal{L}_{\phi}(i)$. When any leader $j\in \mathcal{L}_{\phi}(i)$ has fixed its $\gamma_j^k$ then $i$ requires the best impact $^{j}\bar{\epsilon}_{ij}^{\phi}(t^k)$ as per \eqref{eq: leader best impact} and ${}^{j}\Upsilon^{\phi}_{ij}(t^k)$ as per \eqref{eq: upsilon}. With this information, $\nu_{ij}^{\phi}(t_k) = {}^{i}\Upsilon^{\phi}_{ij}(t^k) + {}^{j}\Upsilon^{\phi}_{ij}(t^k)$ (and thus $\zeta_{ij}^{\phi}$) is computed together with $\tilde{\gamma}^k_{ij}$ as per \eqref{eq:tilde gamma} (lines \ref{lst:line:checking leaders availability}-\ref{lst:line:checking leaders availability 2 }). Note that lines \ref{lst:line:checking leaders availability}-\ref{lst:line:checking leaders availability 2 } in Alg \ref{alg:control reduction alg} can be parallelised if multiple leaders in $\mathcal{L}_{\psi}(i)$ fix their $\gamma_j^k$ simultaneously. When $\tilde{\gamma}^k_{ij}$ has been computed for all the leaders $j\in \mathcal{L}_{\phi}(i)$, the minimum of all $\tilde{\gamma}_{ij}^k$ is taken as $\gamma^k_j$ (line \ref{lst:line:end}). The whole sequential computation of $\gamma_i^k$ requires again  $\lceil \frac{\rho}{2} \rceil$ rounds as  Alg. \ref{alg:leadership propagation}. 

\section{Simulations}\label{simulations}
Consider a MAS composed of a single-integrator system $\dot{\vec{p}}_1= \vec{u}_1 \in \mathbb{R}^2$ and 6 differential drive systems 
$$\begin{bmatrix}\dot{\vec{p}}_i \\ \dot{\theta}_i \end{bmatrix}  = \begin{bmatrix}cos(\theta_i)& l\cdot sin(\theta_i)\\-sin(\theta_i)& l\cdot cos(\theta_i)\\ 0  &1\end{bmatrix} \begin{bmatrix}v_i \\ \omega_i \end{bmatrix} \; \forall i\in 2,\ldots 7 $$
where  $v_i$ and $\omega_i$ are linear and angular velocities respectively and $l>0$ is a small \textit{look-ahead} distance.  We select a sampling interval $\delta t =0.1s$ and communication radius $r_c=8.5$. The control constraints are $\|\vec{u}_1\| \leq 0.9$ and $|v_i|\leq 0.6\land |\omega_i|<0.15,\, \forall i=2,\ldots 7$. We then define the acyclic set of edges $\mathcal{E}_{\psi} = \{(1,2),(1,3),(2,4),(2,5),(3,6),(3,7),(1,1)\}$ with collaborative formation tasks $\varphi_{12}= F_{[30,40]} \|\vec{p}_{12}\| \leq 5, \varphi_{37}= F_{[20,30]} \|\vec{p}_{37} - [2.5,-2.5]^T\|\leq 2, \varphi_{36}= G_{[20,30]} \|\vec{p}_{37} - [-2.5,2.5]^T\|\leq 2$ and communication maintenance task $\varphi_{ij}^{comm},\; \forall (i,j)\in \mathcal{E}_{\psi}$ as per Corollary \ref{communicationc corollary}. Moreover, let the independent task $\varphi_1 = G_{[20,50]}\|\vec{p}_{1} - [15,0]^T\|\leq 3$. A collision avoidance mechanism to avoid 3 static obstacles and one dynamic obstacle (blue circle in Fig. \ref{fig:mas_ca}) is considered in parallel to the control law \eqref{eq:new decentralized optimization with tasks}. The simulation is run for 50 \textit{s} on an Intel-Core i7-1265U with 32 GB of RAM, and the result is shown in Fig. \ref{fig:mas_ca}. The time evolution of the barrier functions $\vec{b}^{\phi}_{ij}(\vec{e}_{ij},t)$ (Fig. \ref{fig:barriers}) shows that all the collaborative tasks are satisfied ($\vec{b}^{\phi}_{ij}(\vec{e}_{ij},t) \geq 0\; \forall (i,j)\in \mathcal{E}_{\psi}$) while the independent task for Agent $1$ is unsatisfied with a delay of 20s with respect to the required time of satisfaction due to the prioritization of the collaborative tasks. The value of the functions $\nu_{ij}^{\phi},\nu_{i}^{\phi}$ is shown in Fig. \ref{fig: nu}, while Fig. \ref{fig:prop_time} shows the computational time required to propagate control reduction factor $\gamma_i^k$ from the leaf nodes $\{5,6,7,4\}$ to the root agent $1$ as per Alg. \ref{alg:control reduction alg}. This time includes the online computation of $\nu^\phi_{ij}$,$\nu^\phi_{i}$ as well as the computation of best/worst impacts $\bar{\epsilon}^{\phi}_{ij}$,$\underbar{\text{$\epsilon$}}^{\phi}_{ij}$ at each $t^k$. Assuming zero communication time, the average propagation time is $6.9$ msec with a standard deviation of 1 msec. In Fig. \ref{fig:gamma} the value of $\gamma_i^k$ is shown at each $t^k$ where for most of the agents, the value remains at $\gamma_i^k=1$ while agents $1$ and $3$ have a minimum control reduction factor of $0.67$ and $0.51$ respectively at times $15.5s$ and $0s$.

\begin{figure*}[ht]
\centering
\begin{subfigure}{.3\textwidth}
  \centering
  \includegraphics[width=1.\columnwidth]{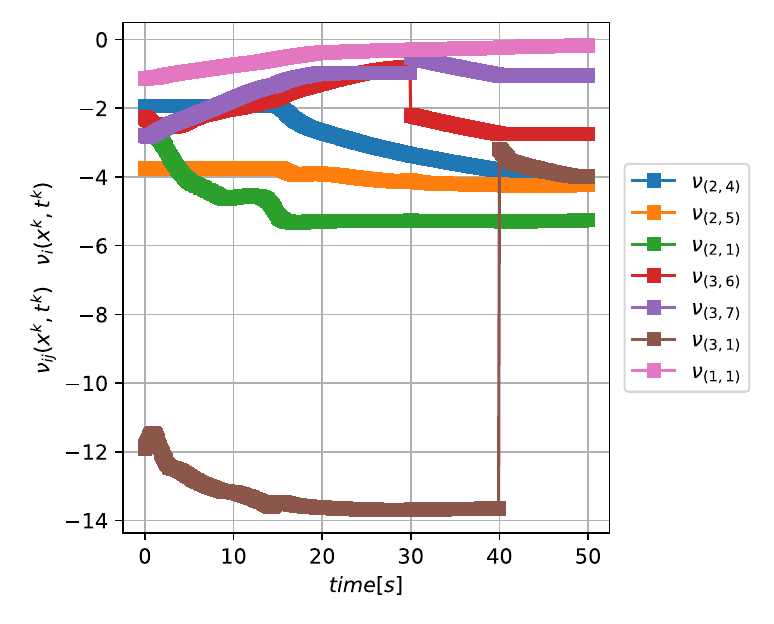}
  \caption{}
  \label{fig: nu}
\end{subfigure}%
\begin{subfigure}{.3\textwidth}
  \centering
  \includegraphics[width=1.\columnwidth]{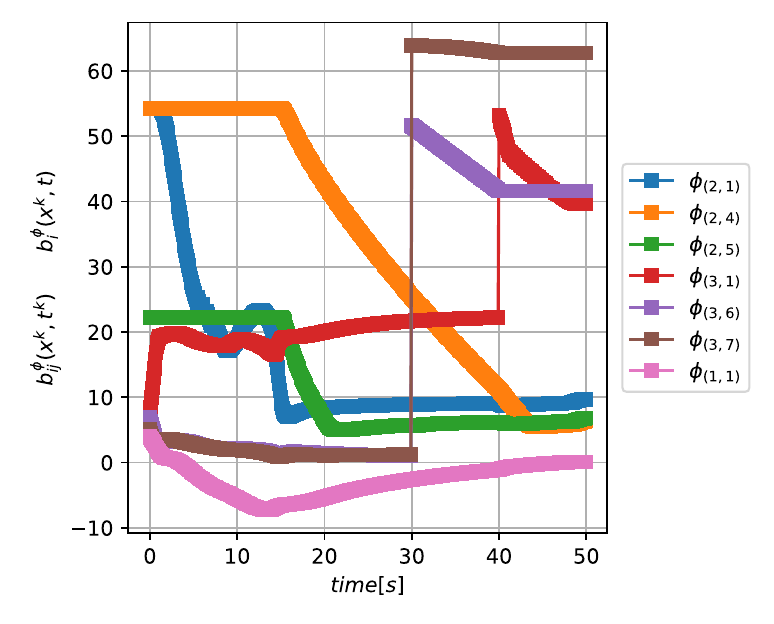}
  \caption{}
  \label{fig:barriers}
\end{subfigure}%
\begin{subfigure}{.3\textwidth}
  \centering
  \includegraphics[width=1.\columnwidth]{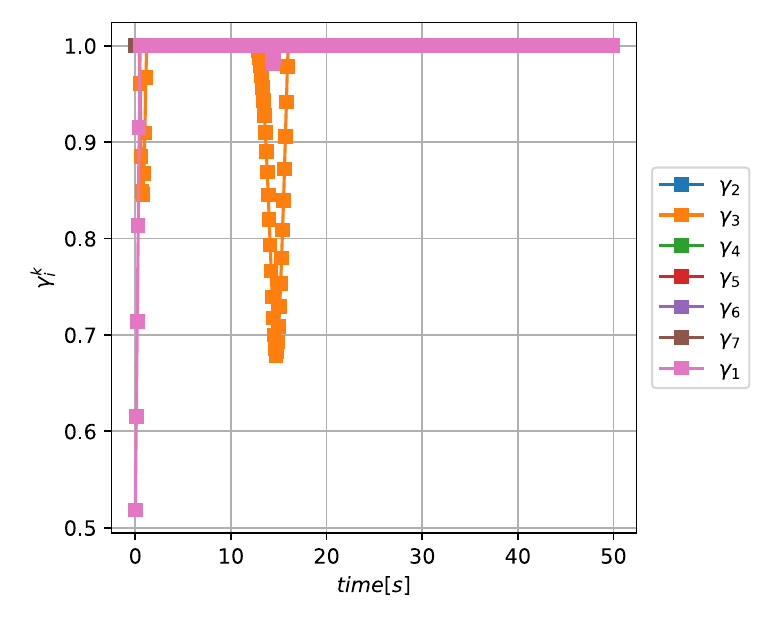}
  \caption{}
  \label{fig:gamma}
\end{subfigure}
\vspace{-0.8cm}\\
\begin{subfigure}{.3\textwidth}
  \centering
\includegraphics[width=.8\columnwidth]{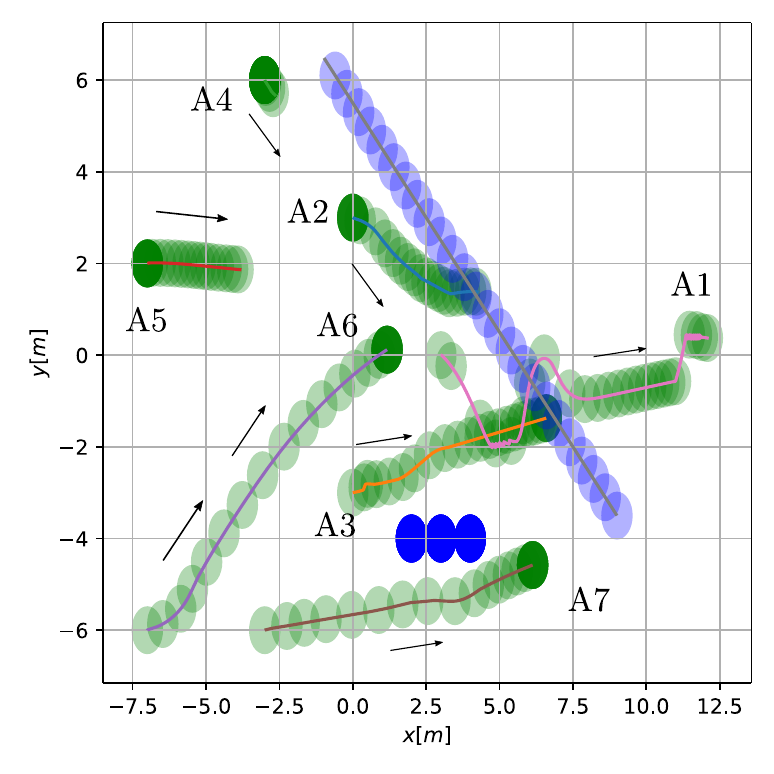}
  \caption{}
  \label{fig:mas_ca}
\end{subfigure}%
\begin{subfigure}{.3\textwidth}
  \centering
  \includegraphics[width=0.94\columnwidth,trim={0cm 2cm 0 0},clip]{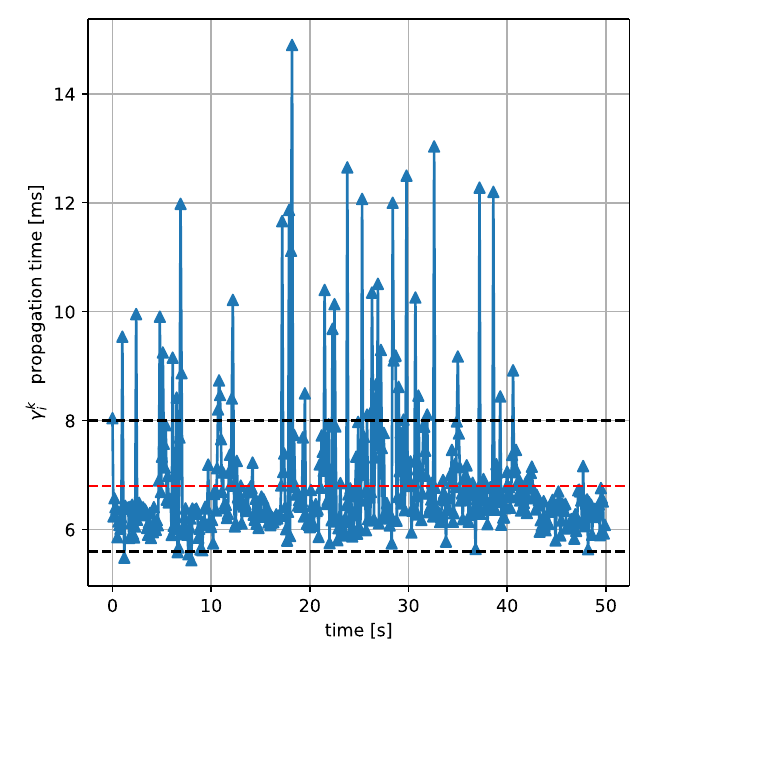}
  \caption{}
  \label{fig:prop_time}
\end{subfigure}%
\caption{ (a) Margin functions $\vec{\nu}^{\phi}_{ij}(\vec{e}_{ij}^k,t^k)$,$\vec{\nu}^{\phi}_{i}(\vec{x}_{i}^k,t^k)$ for each barrier $b_{ij}^{\phi}(\vec{e}_{ij},t), b_{i}^{\phi}(\vec{x}_{i},t)$. (b) Barrier functions  $b_{ij}^{\phi}(\vec{e}^k_{ij},t^k),b_{i}^{\phi}(\vec{x}^k_{i},t^k)$ associated with tasks $\phi_{ij},\phi_i$. (c)  Control reduction factor $\gamma_i^k$. (d) MAS simulation. The arrows represent directions of motion. Blue and green circles represent obstacles and agents respectively (e) Time required to compute $\gamma_i^k$ sequentially for each $i\in \mathcal{V}$ at each $k$ (Alg. \ref{alg:control reduction alg}).}
\label{fig:mas sym big}
\end{figure*}

\section{Conclusions}\label{conclusions}
We showed a novel decentralised control approach for multi-agent systems subject to spatio-temporal and communication constraints exploring the notion of acyclic task graph. Our framework is amenable to direct implementation into sampled-data real embedded controllers with continuous time guarantees over task satisfaction.  In future work, we will explore the robustness and scalability of the proposed approach.


\bibliographystyle{ieeetr} 
\bibliography{references} 

\end{document}